\documentclass[onecolumn]{IEEEtran}
\usepackage{amsmath, amssymb, amsthm, latexsym,hyperref}

\def\0{{\mathbf 0}}
\def\1{{\mathbf 1}}

\def\FF{{\mathbb F}}

\def\Z{{\mathbb Z}}

\def\Zp8{{\Z_{p^\infty}}}

\newtheorem{thm}{Theorem}[section]
\newtheorem{prop}[thm]{Proposition}
\newtheorem{obs}[thm]{Observation}
\newtheorem{lem}[thm]{Lemma}
\newtheorem{cor}[thm]{Corollary}
\theoremstyle{definition}

\newtheorem{rem}[thm]{Remark}

\date{\today}

\begin{document}

\newcommand{\comment}[1]{} 

\title{The Classification of Complementary Information Set Codes of Lengths $14$ and $16$}

\author{
Finley Freibert
}
\maketitle

\begin{abstract}

In the paper
``A new class of codes for Boolean masking of cryptographic computations
,'' Carlet, Gaborit, Kim, and Sol\'{e} defined a new class of rate one-half binary codes called \emph{complementary information set} (or CIS) codes.  The authors then classified all CIS codes of length less than or equal to 12.  CIS codes have relations to classical Coding Theory as they are a generalization of self-dual codes. As stated in the paper, CIS codes also have important practical applications as they may improve the cost of masking cryptographic algorithms against side channel attacks.  In this paper, we give a complete classification result for length 14 CIS codes using an equivalence relation on $GL(n,\FF_2)$.  We also give a new classification for all binary $[16,8,3]$ and $[16,8,4]$ codes.  We then complete the classification for length 16 CIS codes and give additional classifications for optimal CIS codes of lengths 20 and 26.

\end{abstract}





\section{Motivations}

A generalization of self-dual codes was recently proposed by Carlet, Gaborit, Kim, and Sol\'{e} in~\cite{CIS}.  In the paper, a new class of codes, called \emph{complementary information set} (or \emph{CIS}) codes, is defined.  Given an integer $n$, a binary linear code with parameters $[2n,n,d]$ which has two disjoint information sets is a \emph{complementary information set} code.  CIS codes have a variety of connections and applications; the authors (in~\cite{CIS}) note the direct applications found in Cryptography, with relations to Boolean S-Boxes, Boolean functions, and masking~\cite{CIS_2,CIS_3,CIS_1,CIS_4}.  

Previous results on the classification of rate one-half codes date back to Pless' enumeration of self-dual codes in 1972~\cite{Ple_72}.  Since that time many classification results for self-dual codes have been obtained, most recently the doubly-even self-dual codes of length 40~\cite{Har2011}.
A related problem is the classification of formally self-dual codes.  The main results in this direction are classifications of optimal formally self-dual codes; results have been given in~\cite{BetHara, BetHara2, FieGabHuffPless, Han_FSD}.  Betsumiya and Harada gave a complete classification for even formally self-dual codes up to length 16~\cite{BetHara2}; Han, H. Lee, and Y. Lee gave a complete classification for odd formally self-dual codes up to length 14~\cite{Han_FSD}.
Some general results on optimal rate one-half codes were obtained by Gulliver and \"{O}stergard in~\cite{GulOst}.
In the paper~\cite{CIS}, CIS codes are classified for $2n=2,4,6,8,10,12$.  In the proceeding sections we obtain new results in the direction of these previous researchers.

Notations are introduced in Section~\ref{nota}.  Graph isomorphism tools used for the classifications are described in Sections~\ref{classtool}-\ref{GLGraph}.  In Section~\ref{CIS14}, we classify all $[14,7]$ CIS codes (i.e. the case where $n=7$).  In Section~\ref{CIS16_4}, we give new results on the classification of $[16,8,3]$ and $[16,8,4]$ codes which yields a classification of both $[16,8,3]$ and $[16,8,4]$ CIS codes and $[16,8,3]$ and $[16,8,4]$ odd formally self-dual codes.  In this section we develop an algorithm to decide whether a code is CIS.  In Section~\ref{restCIS2}, some theoretical restrictions on CIS codes with minimum weight 2 are stated.  This allows for the complete classification of length 16 CIS codes.  In Section~\ref{CISoptsection}, we give an up-to-date classification of optimal CIS codes. All computations were completed using MAGMA~\cite{mag}; some computations were run in parallel and then the output data was compiled at the end.  Generator matrices for length 14 CIS codes will be posted to the author's website~\cite{results}.

\section{Notations}\label{nota}
The main notations and basic definitions concerning linear codes are adapted from~\cite{HuffPless}.  Let  $\FF_2$ denote the binary field.  Any subspace $C$ of the vector space  $\FF_2^n$ is called a \emph{linear $[n,k]$ code}.  All codes we refer to are both binary and linear.  
The \emph{Hamming weight} of a vector  ${\bf x} \in \FF_2^n$, denoted $wt({\bf x})$ is the number of nonzero coordinates of ${\bf x}$.  The \emph{Hamming distance} between two vectors ${\bf x},{\bf y} \in \FF_2^n$, denoted $d({\bf x},{\bf y})$ is the number of coordinates in ${\bf x}$ and ${\bf y}$ which are different.  The \emph{minimum weight} of a code $C$ is an integer $d$, where the minimum is taken among all non-zero weights of $C$, $C$ is then referred to as an $[n,k,d]$ code.
Two binary codes are said to be \emph{equivalent} if there exists a permutation of coordinates mapping one code onto the other code.
 
Given two vectors in $\FF_2^n$, $v=a_1 a_2 \dots a_n$ and $u=b_1 b_2 \dots b_n$, the \emph{Euclidean inner product} of ${\bf u}$ with ${\bf v}$ is the sum ${\bf u} \cdot {\bf v} := a_1 b_1 + a_2 b_2 \cdots + a_n b_n$.   The \emph{dual of $C$} is the set $C^\perp := \{ {\bf x} \in \FF_2^n : {\bf x} \cdot {\bf v} = 0$ for all ${\bf v} \in C \}$.  The number of codewords of weight $w$ in a code $C$ (resp. $C^\perp$) is denoted by $A_w$ (resp. $A^\perp_w$).  $C$ is called \emph{self-dual} if $C=C^\perp$.  $C$ is called \emph{formally self-dual} if $A_w = A^\perp_w$ for all weights $w$.

A $k$ by $n$ matrix $G$ is called a \emph{generator matrix} for an $[n,k]$ code $C$ if the rows of $G$ form a basis for $C$.  Any set of $k$ columns of $G$ which are linearly independent is called an \emph{information set} for $C$.  As stated above, the authors in~\cite{CIS} define a \emph{complementary information set} code to be a binary linear code with parameters $[2n,n,d]$ which has two disjoint information sets.

\section{A Classification Tool Using Graph Isomorphism}\label{classtool}

The classification of binary $[n,k,d]$ codes satisfying various properties is a classical problem; as mentioned above, previous work in this direction includes the classification of self-dual codes, formally self-dual codes, and rate one-half codes in general.  Thus, an interesting problem in the area of CIS codes is the classification problem.  One main difficulty that arises when classifying codes is the equivalence test.  When comparing a small set of codes the equivalence test can be implemented easily (in MAGMA~\cite{mag}) by performing a pairwise comparison of all codes in the set.  However, when comparing more than a few thousand codes the test becomes rather time consuming.  In essence this is a combinatorial problem of classifying objects up to a defined equivalence.  A useful solution for this problem, proposed independently in 1978 by~\cite{OrdGen1,OrdGen2}, is to generate a list of inequivalent combinatorial objects (codes) by producing a ``canonical representative'' for each equivalence class.  This method is described by Kaski and  \"{O}stergard and it is called \emph{Orderly Generation} (\cite{KasOst} pp.120-124).  There is no equivalence test in this method, the only criterion is set membership.

The difficulty in applying the \emph{Orderly Generation} method is finding a way to determine a ``canonical representative'' for each equivalence class. As suggested in~\cite{KasOst}, a clever navigation of this difficulty is to make use of Brendan McKay's graph isomorphism program \emph{nauty}~\cite{McKay}.  Two graphs $G$ and $G'$ with vertex sets $V$ and $V'$ are said to be \emph{isomorphic} if there exists a bijection $\phi: V \rightarrow V'$ such that $(u,v)$ is an adjacent pair of vertices in $G$ if and only if $(\phi(u),\phi(v))$ is an adjacent pair of vertices in $G'$.  Given a graph $G$ with vertex set $V$ and a fixed labeling on the vertices with the integers ${1,2,...,|V|}$, \emph{nauty} can output a ``canonical'' labeling among all isomorphic graphs.  In fact, if the graph is a colored graph, then \emph{nauty} will give a canonical labeling which preserves the color among labels.  In~\cite{Ost},  \"{O}stergard uses \emph{nauty} functionality to classify binary linear codes of minimum distance greater than two for up to length 14.  In~\cite{Schaathun}, Schaathun implements a search which classifies all $[36,8,16]$ linear codes using \emph{nauty}.

\section{A Correspondence Between Codes and Graphs}\label{CodeGraph}

Now we must describe how to transform a linear code to a colored graph.  As per the formulations in~\cite{KasOst,Ost,Schaathun}, let a linear $[n,k,d]$ code $C$ be given.  Let $S$ be the set of minimum weight in $C$.  If $S$ does not generate $C$, then include all codewords in $C$ of weight 1 higher than the maximum weight in $S$.  Repeat the last step until $S$ generates $C$.  Fix an ordering on $S$ so that $c_i$ represents a specific element of $S$ for $i \in \{1,...,|S|\}$.  Construct a set of $|S|+n$ vertices labeled with the integers $1,2,...,|S|+n$ (denote $v_i$ the vertex with label $i$).  Construct a bipartite graph in the following way.  Let $\{v_1, v_2, ..., v_{|S|} \}$ be one partite set, and let the other partite set be $\{v_{|S|+1}, v_{|S|+2}, ..., v_{|S|+n} \}$.  Draw an edge $(v_i, v_{|S|+j}$ if and only if $c_i$ has a $1$ in coordinate $j$.  Color vertices $\{v_1, v_2, ..., v_{|S|} \}$ black.  Color vertices $\{v_{|S|+1}, v_{|S|+2}, ..., v_{|S|+n} \}$ red.
The following lemma is adapted from the known methods described in~\cite{KasOst,Ost,Schaathun}.
\begin{lem}
A permutation $\alpha _1$ of the labels on the black vertices corresponds to a permutation of the ordering on the codewords.  A permutation $\alpha _2$ of the labels of the red vertices corresponds to a permutation of columns of codewords.  As a result, applying $\alpha _1$ and $\alpha _2$ to a graph $G$ (constructed from a code $C'$), yields a graph $G'$ (corresponding to a code $C'$ equivalent to $C$).
\end{lem}
\begin{proof}
The first claim is clear from the construction since $c_i$ corresponds to vertex $v_i$.  The second claim follows from the fact that if $\alpha _2 (v_{|S|+i})=v_{|S|+j}$, then all codewords which had a 1 in column $i$, now have a 1 in column $j$ after applying $\alpha _2$.

Since $\alpha _1$ and $\alpha _2$ correspond to permuting generators and columns in the code $C$ to obtain $G'$, then $G'$ must correspond to a code $C'$ equivalent to $C$.
\end{proof}

Because of the functionality in \emph{nauty}, a canonically labeled graph (with the color restriction described above) corresponds to a canonical form of a linear $[n,k,d]$ code.  Therefore we may apply the \emph{Orderly Generation} method.  

\section{A Correspondence Between $GL(n,\FF_2)$ and Graphs}\label{GLGraph}
The \emph{general linear group} of degree $n$, denoted $GL(n,\FF_q)$, is the set of all $n$ by $n$ invertible matrices, over $\FF_q$, under matrix multiplication.
Given any linear $[2n,n,d]$ code $C$, it is clear that if the coordinate set $\{1,2,...,n\}$ forms an information set, then any generator matrix of $C$ has the form $G=[I | A]$, after performing Gaussian Elimination, where $I$ is the $n$ by $n$ identity matrix and $A$ is an $n$ by $n$ matrix.  In~\cite{CIS} this is called the \emph{systematic form} of the generator matrix for a $[2n,n,d]$ code $C$. $C$ is CIS if and only if $C$ may be converted to \emph{systematic form} where $A \in GL(n,\FF_2)$, by Lemma IV.1 of~\cite{CIS}.  Hence if the equivalence classes of $GL(n,\FF_2)$ are classified, then the classification of CIS codes can be obtained using the ideas of Section~\ref{CodeGraph}.  Therefore an interesting related classification problem is to find all equivalence classes of $GL(n,\FF_2)$ (under row and column permutations).

Consider the following notion of equivalence on $GL(n,\FF_2)$: two matrices $A,B \in GL(n,\FF_2)$ are equivalent, $A\sim B$, if and only if $A=P_1 B P_2$ if $P_1$ and $P_2$ are two $n$ by $n$ permutation matrices.
\begin{obs}
$\sim$ is an equivalence relation on $GL(n,\FF_2)$.
\end{obs}
\begin{proof}
Let $A,B,C \in GL(n,\FF_2)$.  The notation $P_i$ for an integer $i$ denotes a permutation matrix.
Reflexivity is clear since if $P_e$ is the identity permutation, then $A=P_e A P_e$.
Symmetry holds since the inverse of a permutation matrix is a permutation matrix: $A=P_1 B P_2$ implies $B= P_1^{-1} A P_2^{-1}$.
Transitivity holds since the product of two permutation matrices is a permutation matrix: $A=P_1 B P_2$ and $B=P_3 C P_4$ implies $A=(P_1 P_3) C (P_4 P_2)$.
\end{proof}

\begin{obs}\label{obs:gl}
If $A,B \in GL(n,\FF_2)$ are such that $A \sim B$, then the CIS codes with systematic generator matrices $[I|A]$ and $[I|B]$ are equivalent.
\end{obs}
\begin{proof}
By definition, $A=P_1 B P_2$ where $P_1$ and $P_2$ are two $n$ by $n$ permutation matrices.  Let $B'=BP_2$.  The code generated by $P_1[I|B']$ is equivalent to the code generated by $[I|B]$ by permuting the columns of B and the rows of the entire generator matrix.  However, $P_1[I|B']=[P_1|(P_1BP_2)]=[P_1|A]$ since $P_1 I = P_1$.  Finally to obtain $[I|A]$ from $[P_1|A]$ simply permute the columns of $P_1$ by applying $P_1^{-1}$.
\end{proof}

\begin{rem}
The converse of Observation~\ref{obs:gl} is not true in general.  The smallest counterexample is for $n=3$.  Consider the following matrices from $GL(3,\FF_2)$ 
\[
A = \left[ \begin{array}{c}
111 \\
011 \\
001
\end{array}
\right],
B = \left[ \begin{array}{c}
110 \\
011 \\
001
\end{array}
\right].
\]
Note that $A$ is not equivalent to $B$ since $A$ has a row of weight 3 and $B$ does not.  However, it is clear to see that the code generated by $[I|A]$ is equivalent to the code generated by $[I|B]$ by permuting the second and fifth columns and performing row elimination.
\end{rem}

To use \emph{nauty}, we now describe how to transform an element of $GL(n,\FF_2)$ to a colored bipartite graph.  Similar to the method of Section~\ref{CodeGraph}, let $A \in GL(n,\FF_2)$.  Construct a set of $2n$ vertices labeled with the integers $1,2,...,2n$ (denote $v_i$ the vertex with label $i$).  Construct a bipartite graph in the following way.  Let $\{v_1, v_2, ..., v_{n} \}$ be one partite set, and let the other partite set be $\{v_{n+1}, v_{n+2}, ..., v_{2n} \}$.  Draw an edge $(v_i, v_{n+j}$ if and only if row $i$ has a $1$ in column $j$.  Color vertices $\{v_1, v_2, ..., v_{n} \}$ black.  Color vertices $\{v_{n+1}, v_{n+2}, ..., v_{2n} \}$ red.
The following lemma is adapted from the known combinatorial formulations in~\cite{KasOst}.
\begin{lem}
A permutation $\alpha _{row}$ (resp. $\alpha _{col}$) of the labels on the black (resp. red) vertices corresponds to a permutation of rows (resp. columns).  As a result, applying $\alpha _{row}$ and $\alpha _{col}$ to a graph $G$ (constructed from $A \in GL(n,\FF_2)$), yields a graph $G'$ (corresponding to an equivalent matrix $A' \in GL(n,\FF_2)$).
\end{lem}
\begin{proof}
The first claim follows from the construction since a row corresponds to a vertex in $\{v_1, v_2, ..., v_{n} \}$ and a column position corresponds to a vertex in $\{v_{n+1}, v_{n+2}, ..., v_{2n} \}$.

Since $\alpha _{row}$ and $\alpha _{col}$ correspond to permuting rows and columns in the matrix $A$ to obtain $G'$, then $G'$ must correspond to a matrix $A'$ equivalent to $A$.
\end{proof}

\section{$[14,7]$ CIS Codes}\label{CIS14}

For length 14 an optimal CIS code is mentioned in~\cite{CIS}; this code is self-dual with parameters $[14,7,4]$.  In order to apply the theories developed in the previous section we need a construction method for the elements of $GL(n,\FF_2)$.
Our aim in this section is to first classify elements (up to equivalence) in $GL(n,\FF_2)$ for $n \leq 7$, then we use these elements to classify all CIS codes of length 14.

We first discuss how to obtain matrices in $GL(n,\FF_2)$ using inequivalent matrices from $GL(n-1,\FF_2)$.
The following two lemmas are adapted from Lemma VI.3 and Proposition VI.4 of~\cite{CIS}.

\begin{lem}
Any matrix $A \in GL(n,\FF_2)$ has a submatrix $A' \in GL(n-1,\FF_2)$.
\end{lem}
\begin{proof}
Let $a_i$ be the $i$th column of $A$ and let $r_i$ be the $i$th row of $A$ where $1 \leq i \leq n$.  Delete $a_1$ from $A$ to obtain an $n$ by $n-1$ matrix $A_1$.  Let $r_i '$ be the $i$th row of $A_1$.  Since $A_1$ has rank $n-1$, there exists a $j$ such that $\{ r'_i : i \neq j\}$ are linearly independent and $r_j '= \displaystyle\sum _{i \neq j} c_i r'_i$ for uniquely determined $c_i$.  Therefore by deleting $r'_j$ from $A_1$ we obtain an $n-1$ by $n-1$ matrix $A'$ having rank $n-1$.
\end{proof}

\begin{lem}
For any matrix $A' \in GL(n-1,\FF_2)$, a matrix $A \in GL(n-1,\FF_2)$ may be obtained by the following:
For any $x,y \in \FF_2^{n-1}$, fix $c:=x A^{-1}$ and $z:=[1] +c y^T$, then

{
\[ A=\left[ \begin{array}{cc}
z & x \\
y^T &A'
\end{array}
\right]
\]
}.
\end{lem}
\begin{proof}
Since the rows of $A'$ are linearly independent $x$ must be a linear combination of the rows of $A'$, which implies there exists a $c \in \FF_2^{n-1}$ such that $c A = x$.  Solving for $c$ we obtain $c=x A^{-1}$.  To ensure that the top row of $A$ is linearly independent from the other rows the value of $z$ must be such that $c [y^T A' ] \neq [z$ $x]$.  Hence $c y^T \neq z$, and as the values are binary this is equivalent to $c y^T + [1] = z$.
\end{proof}

By applying this theory recursively to all representatives from equivalence classes of $GL(n-1,\FF_2)$ along with the canonical selection method in Section~\ref{GLGraph} we may obtain all equivalence class representatives in $GL(n,\FF_2)$.  For $n=1, 2, ... , 7$ we have obtained the number of equivalence classes given in Table~\ref{tab:eq_class}.

{
\begin{table}[h!tb]
 \caption{Number of Equivalence Classes in $GL(n,\FF_2)$ Under Row \& Column Permutations}
 \label{tab:eq_class}
 \begin{center}
{
\small
\begin{tabular}{|r|c|c|c|c|c|c|c|}
\noalign{\hrule height1pt}
$n=$ & 1 & 2 & 3 & 4 & 5 & 6 & 7  \\ 
\hline
Total & 1 & 2 & 7 & 51 & 885 & 44,206 & 6,843,555  \\ 

 \noalign{\hrule height1pt}
\end{tabular}
}
\end{center}
\end{table}
}

{
\begin{table}[h!tb]
 \caption{Classification of Length 14 CIS codes}
 \label{tab:class14}
 \begin{center}
{
\small
\begin{tabular}{|r|c|c|c|c|}
\noalign{\hrule height1pt}
& Total CIS & SD & Only FSD & Not SD or FSD  \\ 
\hline
$d=2$ & 62,015 & 3 & 4,407 & 57,605  \\ 

$d=3$ & 22,561 & 0  & 2,160 & 20,401  \\ 

$d=4$ & 1,476 & 1~\cite{CIS} & 121 & 1,354  \\ 
\hline
Total & 86,052 & 4 &  6,688 & 79,360  \\ 


 \noalign{\hrule height1pt}
\end{tabular}
}
\end{center}
\end{table}
}

For each representative $A$ from equivalence classes of $GL(n,\FF_2)$, appending the $n$ by $n$ identity matrix $I$, $[I |A]$ is a generator matrix for a CIS code.  By applying the method introduced in Section~\ref{CodeGraph} we can then obtain a set of all inequivalent CIS codes of length $2n$.  Hence we obtain the following classification theorem.

\begin{thm}
There are exactly 86,052 $[14,7]$ CIS codes.
\end{thm}

Additional information for the $[14,7]$ CIS codes is listed in Table~\ref{tab:class14}, the rows give the possible minimum distances and the columns tell how many are self-dual, formally self-dual but not self-dual, and neither.

\section{The $[16,8,3]$ and $[16,8,4]$ Codes}\label{CIS16_4}

For length 16 an optimal CIS code is mentioned in~\cite{CIS}; this code with parameters $[16,8,5]$ was shown to be unique by Betsumiya and Harada and in fact this code is formally self-dual~\cite{BetHara}.
A classification for length 16 CIS codes would be interesting, but since the number of equivalence classes of $GL(7,\FF_2)$ is very large, it is not feasible to determine the classes of $GL(8,\FF_2)$.  Hence the construction from the previous section can not be used.
Therefore we consider another method for determining the $[16,8,3]$ and $[16,8,4]$ CIS codes.  The method we use is to generate all binary linear $[16,8,3]$ and $[16,8,4]$ codes and then determine which ones are CIS.  To ease the computation we make use of the well known fact that any $[n,k]$ odd code contains a unique $[n,k-1]$ subcode generated by all even codewords.  Also, it is clear that any $[16,7]$ subcode of any $[16,8]$ even code is even.  Therefore we first classify all $[16,7,\ge 4]$ even codes.
We give the following lemma based on the theory of shortening codes in~\cite{HuffPless} to justify our method for finding the $[16,7,\ge 4]$ even codes.
\begin{lem}
If $C$ is a binary $[n,k]$ code with generator matrix in standard form $G$, then shortening $C$ on the first column yields an $[n,k-1]$ code. 
\end{lem}
\begin{proof}
Since $G$ is in standard form the only row of the generator matrix with a 1 in the first column is the first row.  Therefore, shortening on the first column yields an $[n,k-1]$ code.
\end{proof}

Applying this lemma recursively to any $[n,k,d]$ code, a nested chain of subcodes is obtained, the smallest subcode having parameters $[n-k+1,1, \geq d]$.  Therefore, any $[16,7,\geq 4]$ code has a nested chain of subcodes (``subcode'' meaning by adding a zero column it is a subcode):
\[
[16,7,\geq 4] \supset [15,6,\geq 4] \supset [14,5,\geq 4] \supset [13,4,\geq 4] 
\]
\[
\supset [12,3,\geq 4] \supset [11,2,\geq 4] \supset [10,1,\geq 4]
\]
If we have a list of all inequivalent $[n',k', \geq 4]$ codes $L$ we construct all $[n'+1,k'+1, \geq 4]$ supercodes by adding a zero column onto each code $C$ in $L$ and then increasing the dimension by adding vectors from $\FF_2^{n'+1} / C$.  This method is somewhat opposite from the method described in~\cite{Ost} which instead adds columns to the parity check matrix.
We apply the method recursively and keep only ``canonical'' representatives as in Section~\ref{CodeGraph} to obtain a classification of 29,243 total inequivalent $[16,7,\geq 4]$ even codes.  29,240 of these codes are $[16,7,4]$ and 3 of them are $[16,7,6]$ (the $[16,7,6]$ codes were previously classified by Simonis in~\cite{Simonis2}).

Our goal is to classify all $[16,8,3]$ and $[16,8,4]$ codes, so the next step is for each $[16,7,\geq 4]$ even code $C$ we add all possible vectors $x$ from $\FF_2^{16} / C$ to form codes $C+ <x>$.  We keep a list of all $[16,8,\geq 3]$ codes generated in this way.  To determine which codes are inequivalent we keep only the ``canonical'' representatives as in Section~\ref{CodeGraph}.  Our conclusion from this search is the following theorem.

\begin{thm}\label{16codes}
There are exactly 2,914,299 binary $[16,8,3]$ codes and there are exactly 271,783 binary $[16,8,4]$ codes.
\end{thm}

 In the Tables~\ref{tab:class16_8_3} and~\ref{tab:class16_8_4} we have the totals for how many of these codes are self-dual, only even formally-self-dual, only odd formally self-dual, and neither self-dual nor formally self-dual; there we note the previously classified self-dual codes~\cite{Ple_72} and even formally self-dual codes~\cite{BetHara2}.  We also include a column which states how many have $d^\perp \neq 1$, which means there are no zero columns in the generator matrix.  

Since we have a list of all inequivalent $[16,8,3]$ and $[16,8,4]$ codes we may then pursue the main goal of classifying the ones which are CIS.  To determine if a code is CIS we use the following algorithm.

\medskip

\noindent
{\bf CIS Determination Algorithm:}  An algorithm to determine if a given code is CIS.

\begin{enumerate}

\item Input:
Begin with a binary $[2n,n]$ code $C$.

\item Output: An answer of ``Yes'' if $C$ is CIS and ``No'' if not.

\medskip

\begin{enumerate}
\item Fix a generator matrix $G$ of $C$.  Initialize an empty set to hold used integer subsets $U:=\{ \}$.
\item Fix the first column of $G$ by initializing $I:=\{1\}$ ($I$ holds the columns of an information set we are building).
\item Choose $n-1$ integers $\{ i_1,...,i_{n-1} \}$ from the set $\{2,3,\dots, 2n-1,2n\}$ such that $\{ i_1,...,i_{n-1} \} \notin U$.  Include $\{ i_1,...,i_{n-1} \}$ as an element in $U$.  Let $I:=I \cup \{ i_1,...,i_{n-1} \}$.
\item If the columns of $G$ indexed by $I$ are not linearly independent or indexed by $\{1,2,3,\dots, 2n-1,2n\} \setminus I$ are not linearly independent, then go to (e).  Otherwise, if the columns of $G$ indexed by $I$ are linearly independent and  indexed by $\{1,2,3,\dots, 2n-1,2n\} \setminus I$ are linearly independent, then output ``Yes'' and exit algorithm.
\item If $|U| = {\binom{2n-1}{n-1}}$, then output ``No'' and exit algorithm; otherwise, go back to (b).
\end{enumerate}

\end{enumerate}

This algorithm searches through all possible information sets of size $n$ and checks if the complement is also an information set.  The first column may be fixed in step (b) since without loss of generality the first column appears in one information set of size $n$ in any CIS code.  This algorithm searches through at most ${\binom{2n-1}{n-1}}$ possible information sets; if the code is determined to be CIS the algorithm is exited early in step (d).

To examine a general $[16,8]$ code to determine it is not CIS, the algorithm will have to search ${\binom{15}{7}} = 6435$ information sets.  For a single code this takes approximately 2.129 seconds.  However, since there are 2,914,299+271,783= 3,186,082 codes to examine, applying the algorithm one code at a time would take at most approximately 78.5 days.  Instead we applied the algorithm in parallel to separate codes to decrease the execution time.
By applying this algorithm to all $[16,8,3]$ and $[16,8,4]$ codes we obtain the following conclusion.

\begin{thm}
There are exactly 2,711,027 $[16,8,3]$ CIS codes and there are exactly 267,442 $[16,8,4]$ CIS codes.
\end{thm}

\begin{table}[h!tb]
 \caption{Classification of $[16,8,3]$ codes and $[16,8,3]$ CIS codes}
 \label{tab:class16_8_3}
 \begin{center}
{\small
\begin{tabular}{|r|c|c|c|c|}
\noalign{\hrule height1pt}
& Total & $d^\perp \neq 1$ & Odd FSD & Not FSD  \\ 
\hline
All $[16,8,3]$ & 2,914,299 &  2,780,328 &  162,423 & 2617905\\ 
\hline
CIS $[16,8,3]$ & 2,711,027 & 2,711,027 & 162,406 & 2,548,621 \\ 
 \noalign{\hrule height1pt}
\end{tabular}
}
\end{center}
\end{table}

\begin{table}[h!tb]
 \caption{Classification of $[16,8,4]$ codes and $[16,8,4]$ CIS codes}
 \label{tab:class16_8_4}
 \begin{center}
{\small
\begin{tabular}{|r|c|c|c|c|c|c|}
\noalign{\hrule height1pt}
& Total & $d^\perp \neq 1$ & SD & Only Even FSD & Odd FSD & Not SD or FSD  \\ 
\hline
All $[16,8,4]$ & 271,783 & 268,261 & 3~\cite{Ple_72} & 141~\cite{BetHara2} & 12,827&255,290 \\ 
\hline
CIS $[16,8,4]$ & 267,442 &  267,442 & 3 & 141 & 12,827 & 254,471\\ 
 \noalign{\hrule height1pt}
\end{tabular}
}
\end{center}
\end{table}

As a tangential result to Theorem~\ref{16codes} we examine the number of odd formally self-dual codes.  The recent results on odd formally self-dual codes are given in~\cite{Han_FSD}, where odd formally self-dual codes are classified for lengths up to 14 and additional results are given on optimal odd formally self-dual codes.  We give a classification of odd formally self-dual $[16,8,3]$ and $[16,8,4]$ codes in the following corollary to Theorem~\ref{16codes}.

\begin{cor}
There are exactly 162,423 odd formally self-dual $[16,8,3]$ codes and there are exactly 12,827 odd formally self-dual $[16,8,4]$ codes.
\end{cor}

\section{Restrictions on CIS Codes of Minimum Distance 2}\label{restCIS2}

The following proposition and its corollaries give restrictions on the structure of the systematic generator matrix of a CIS code with minimum distance 2.  This allows for some interesting theory for constructing CIS codes with minimum distance 2.

\begin{prop}\label{minwt2}
Let $C$ be a $[2n,n,2]$ CIS code with generator matrix in systematic form $G=[I | A]$.  If $x$ is a weight 2 codeword of $C$, then $x$ is a row of $G$.
\end{prop}
\begin{proof}
Suppose to the contrary that $x$ is not a row of $G$.  Let $i_j$ be the $j$th row of $I$ and $a_j$ be the $j$th row of $A$, hence let $i_j a_j$ be the $j$th row of $G$.  Since $x$ is a codeword of $C$ there exists a linear combination $x=\sum_{j=1}^n c_j (i_j a_j)$ where $c_j \in \FF_2$.  Which implies the equation $2 = wt(x) = wt(\sum_{j=1}^n c_j (i_j))+wt(\sum_{j=1}^n c_j (a_j))$.  As $x$ is not a row of $G$ then at least two $c_j$s are nonzero.  However, the support of the $i_j$s do not intersect; so if more than two $c_j$s are nonzero, then the weight of $x$ is greater than two.  Therefore exactly two $c_j$s are nonzero; let these be $c_{j'}$ and $c_{j''}$.  So now $wt(c_{j'}(i_{j'}) + c_{j''}(i_{j''})) =2$, which implies by the above weight equation that $wt(c_{j'}(a_{j'}) + c_{j''}(a_{j''})) =0$.  This is a contradiction since it implies the rows of $A$ are not linearly independent.
\end{proof}

\begin{cor}
If $C$ is a CIS code with minimum weight 2, then all weight 2 codewords have disjoint support.
\end{cor}
\begin{proof}
By Proposition~\ref{minwt2} the weight 2 codewords appear in the generator matrix with systematic form.  If two weight 2 codewords do not have disjoint support, then their corresponding rows in $A$ will not be independent.
\end{proof}

\begin{cor}
Let $C$ be a $[2n,n,2]$ CIS code with generator matrix in systematic form $G=[I | A]$.  If a weight 2 codeword of $C$ has support $\{k_1,k_2\}$ with $k_1 < k_2$, then $k_1 \in \{1,2,\dots, n\}$ and $k_2 \in \{n+1,n+2,\dots, 2n\}$.
\end{cor}
\begin{proof}
By Proposition~\ref{minwt2} the weight 2 codewords appear in the generator matrix with systematic form.  If the claim is not true, then either $I$ or $A$ will have an all zero row which is a contradiction.
\end{proof}

Proposition~\ref{minwt2} and its two corollaries give some insight into the construction of CIS codes of minimum weight 2 and allow us to obtain the following theory which is unique to the construction of CIS codes of minimum weight 2.

\begin{prop}\label{const2}
All $[2n,n,2]$ CIS codes (up to column permutation) can be obtained from a list of all inequivalent CIS codes of length $2n-2$.
\end{prop}
\begin{proof}
Let $C$ be any $[2n,n,2]$ CIS code.  Without loss of generality the first row of a systematic generator matrix $G$ of $C$ is a weight 2 vector.  By Lemma VI.3 and Propositions VI.4 and VI.6 of~\cite{CIS} $G$ has the following form:
\[
G = \left[ \begin{array}{c|c|c|c}
1 & 0 0 \cdots 0 & 1 & 0 0 \cdots 0\\
\hline
\bf{0} & I & y^T & A\\
\end{array}
\right]
\]
where {\bf 0} is the all zero column of length $n-1$, $y \in \FF_2 ^{n-1}$, and $[I|A]$ is a generator matrix in systematic form for a $[2n-2,n-1]$ CIS code $C'$.
It is noted in Remark VI.8 of~\cite{CIS} that using the Building-up Construction on a CIS code $C'$ of length $2n-2$ may produce inequivalent sets of CIS codes when a different systematic partition of $C'$ is used.  However, this problem does not occur in this case since all zeros appear above $I$ and $A$ in the generator matrix $G$; hence permuting columns of $C'$ to obtain a different systematic form $[I|A']$ is the same as permuting the columns of $G$ corresponding to $I$ and $A$.
\end{proof}

For any length $2n$ there is special CIS code that can be constructed with $n$ weight 2 vectors by the following proposition.  It may be noted that this code is self-dual.

\begin{prop}\label{uniqueCIS}
There is a unique CIS code of length $2n$ containing $n$ codewords of weight 2.
\end{prop}
\begin{proof}
This code is a $[2n,n,2]$ CIS code.  The generator in systematic form has all weight 2 rows with disjoint supports which follows from Proposition~\ref{minwt2} and its corollaries.
\end{proof}

By applying Proposition~\ref{const2} to all CIS codes of length 14 from Section~\ref{CIS14}, all $[16,8,2]$ CIS codes may be obtained.  We implemented this in MAGMA to deduce the following theorem.

\begin{thm}
There are exactly 4,798,598 $[16,8,2]$ CIS codes.  
\end{thm}

This completes the classification of length 16 CIS codes.  In the style of Table~\ref{tab:class14} we compile the information from Section~\ref{CIS16_4} and the information on the unique optimal $[16,8,5]$ CIS code (from~\cite{CIS}) to give Table~\ref{tab:class16}.

{
\begin{table}[h!tb]
 \caption{Classification of Length 16 CIS codes}
 \label{tab:class16}
 \begin{center}
{
\small
\begin{tabular}{|r|c|c|c|c|}
\noalign{\hrule height1pt}
& Total CIS & SD & Only FSD & Not SD or FSD  \\ 
\hline
$d=2$ & 4,798,598 & 4 & 150,080 & 4,648,514 \\ 

$d=3$ & 2,711,027 & 0 & 162,406 & 2,548,621  \\ 

$d=4$ & 267,442 & 3 & 12,968 & 254,471  \\ 

$d=5$ & 1~\cite{CIS} & 0 & 1~\cite{CIS} & 0  \\ 
\hline
Total & 7,777,068 & 7 &  325,455 & 7,451,606  \\

 \noalign{\hrule height1pt}
\end{tabular}
}
\end{center}
\end{table}
}

\section{On the Classification of Optimal CIS Codes}\label{CISoptsection}

In~\cite{CIS} the authors gave examples of codes with best known minimum distance which are in fact CIS for lengths 2 through 130. The authors also give an example of a code which is optimal but not CIS.  Therefore an interesting problem is the classification of optimal CIS codes (and the determination of optimal codes which are not CIS).

In Section~\ref{CIS14} we classified the length 14 CIS codes.  In~\cite{GulOst}, it was determined that there exist exactly 1535 optimal $[14,7,4]$ codes.  We reconstruct those codes applying the method described in Section~\ref{CIS16_4} and then examine the ones that are not CIS.  In doing so we obtain the following interesting result.

\begin{prop}\label{dd2}
There are exactly 59 optimal $[14,7,4]$ codes which are not CIS.  47 of these codes have dual distance 1 and the other 12 have dual distance 2.
\end{prop}
\begin{proof}
The 47 with dual distance 1 are not CIS by Proposition IV.5 of~\cite{CIS}.
We determine that the remaining 12 are not CIS by applying the CIS Determination Algorithm from Section~\ref{CIS16_4}.
\end{proof}

The 12 $[14,7,4]$ codes listed in the proposition are optimal rate one-half codes of the smallest length with $d^\perp =2$ which are not CIS.  It is noted in Proposition IV.6 from~\cite{CIS}, that there exists at least one optimal code (with parameters $[34,17,8]$) which is not CIS.  The noted $[34,17,8]$ code has dual distance 1.  In Proposition~\ref{dd2}, we give a first example of optimal codes with dual distance greater than 1 that are not CIS.

Optimal $[20,10,6]$ and $[26,13,7]$ were determined in~\cite{GulOst}.  By applying the CIS Determination Algorithm to all 1682 optimal $[20,10,6]$ codes and all 3 optimal $[26,13,7]$ codes we determined the following result.

{
\begin{table}[h!tb]
 \caption{Classification of Optimal CIS and non-CIS Codes}
 \label{tab:CISopt}
 \begin{center}
{
\small
\begin{tabular}{|c|c|c|c|c|}
\noalign{\hrule height1pt}
 2n & $d_{opt}$~\cite{GulOst} & CIS & not CIS & Total~\cite{GulOst}  \\ 
\hline
2 & 2 & 1~\cite{CIS} & 0 & 1 \\ 
4 & 2 & 2~\cite{CIS} & 1 & 3  \\ 
6 & 3 & 1~\cite{CIS} & 0 & 1  \\ 
8 & 4 & 1~\cite{CIS} & 0 & 1  \\ 
10 & 4 & 4~\cite{CIS} & 0 & 4  \\ 
12 & 4 & 41~\cite{CIS} & 2 & 43 \\
14 & 4 & \bf{1476} & \bf{49} & 1535 \\ 
16 & 5 & 1~\cite{CIS} & 0 & 1 \\ 
18 & 6 & 1~\cite{CIS} & 0 & 1  \\ 
20 & 6 & \bf{1682} & 0 & 1682  \\ 
22 & 7 & 1~\cite{CIS} & 0 & 1  \\
24 & 8 & 1~\cite{CIS} & 0 & 1  \\
26 & 7 & \bf{3} & 0 & 3 \\
28 & 8 & 1~\cite{CIS} & 0 & 1  \\
 \noalign{\hrule height1pt}
\end{tabular}
}
\end{center}
\end{table}
}

\begin{prop}
All 1682 optimal $[20,10,6]$ codes are CIS and all 3 optimal $[26,13,7]$ codes are CIS.
\end{prop}

In Table~\ref{tab:CISopt} we give the results which have been obtained so far in this direction.  Known optimal CIS codes described in~\cite{CIS} are cited in the table.  New results on determining which optimal codes are CIS and not CIS are labeled in bold.  Column $2n$ is the length.  The second column is the optimal minimum distance for any $[2n,n]$ ($n\leq 14$) code determined in~\cite{GulOst}.  The columns 3 and 4 give the number of codes which are CIS and not CIS respectively.  The last column is the total number of optimal codes which was determined in~\cite{GulOst}.

\section{Conclusion}

Many open problems in Coding Theory are concerned with the classifications of codes.  Results in this direction have been obtained for self-dual codes, formally self-dual codes, and rate one-half codes in general.  A recent generalization of self-dual codes, called CIS codes, was proposed in~\cite{CIS} and full classification results were given for lengths up to 12.  In the present paper we complete the classification of CIS codes for length 14 and 16.  We also give a classification of all binary $[16,8,3]$ and $[16,8,4]$ codes which in turn yields new classification results for odd formally self-dual and CIS codes with parameters $[16,8,3]$ and $[16,8,4]$.  In the final section, we complete the classification of optimal CIS codes for lengths 20 and 26.

\bigskip

\noindent{\small{\bf{Acknowledgments:}} The author would like to thank Jon-Lark Kim for the helpful comments and valuable suggestions.}

\bigskip

\bibliographystyle{plain}
\bibliography{CISbib}

\end{document}